\documentclass[11pt,a4paper]{article}

\usepackage{amsmath}
\usepackage{amsfonts}
\usepackage{amsthm}
\usepackage[T1]{fontenc}
\usepackage[latin9]{inputenc}
\usepackage[pdftex]{color,graphicx}

\newtheorem{Theorem}{Theorem}
\newtheorem{Corollary}{Corollary}

\newtheorem{Lemma}{Lemma}
\newtheorem{Remark}{{Remark}}

\newcommand{\LL}{{\mathrm{L}}}

\newcommand{\ds}{{\mathrm{d}}s}
\newcommand{\dx}{{\mathrm{d}}x}
\newcommand{\dy}{{\mathrm{d}}y}
\newcommand{\Id}{{\mathbf{1}}}
\newcommand{\dom}{{\mathrm{dom}~}}

\newcommand{\var}{\varepsilon}

\makeatletter
\newcommand{\subjclass}[2][2010]{%
  \let\@oldtitle\@title%
  \gdef\@title{\@oldtitle\footnotetext{#1 \emph{Mathematics subject classification.} #2}}%
}
\newcommand{\keywords}[1]{%
  \let\@@oldtitle\@title%
  \gdef\@title{\@@oldtitle\footnotetext{\emph{Key words and phrases.} #1.}}%
}
\makeatother

\begin{document}

\title{Absolute continuity and band gaps 
of the spectrum of the Dirichlet Laplacian in periodic waveguides}
\author{Carlos R. Mamani
\quad and \quad Alessandra A. Verri}

\date{\today}

\maketitle 

\begin{abstract}
Consider the Dirichlet Laplacian operator $-\Delta^D$ 
in a periodic waveguide $\Omega$. 
Under the condition that 
$\Omega$ is sufficiently thin, 
we show that its  spectrum $\sigma(-\Delta^D)$ is absolutely continuous (in each finite region).
In addition,
we ensure the existence of at least one gap in $\sigma(-\Delta^D)$
and locate it.
\end{abstract}








\section{Introduction and results}

During the last years the Dirichlet Laplacian operator $-\Delta^D$ restricted to strips (in $\mathbb R^2$) 
or tubes (in $\mathbb R^3$) has been studied under various aspects.
We highlight the particular case where the geometry of these regions are periodic
\cite{bento,borisov,friedcabs,friedlander,sobolev,Yoshi}. In this situation,
an interesting point
is to know under what conditions the spectrum  $\sigma(-\Delta^D)$ is purely absolutely continuous.
On the other hand,
since $\sigma(-\Delta^D)$ is a union of bands, another question
is about the existence of gaps in its structure.

In the case of planar periodically curved strips, 
the absolutely continuity was proved by Sobolev \cite{sobolev}
and the existence and location of band gaps was studied by Yoshitomi \cite{Yoshi}.
The goal of this  paper is to prove  similar results to those in
the three dimensional case. 
In the following paragraphs, we explain the details.

Let $r: \mathbb R \to \mathbb R^3$ be a simple $C^3$ curve in $\mathbb R^3$
parametrized by its arc-length parameter $s$ which possesses an appropriate Frenet frame;
see Section \ref{changecoordinates}.
Suppose that $r$ is periodic, i.e., there exists $L > 0$ and a nonzero vector $u$
so that
$r(s+L)=u+r(s)$, for all $s \in \mathbb R$. Denote by $k(s)$ and $\tau(s)$ the curvature and torsion of $r$ at the position $s$, respectively.
Pick $S \neq \emptyset$; an open, bounded, smooth and connected subset of~$\mathbb R^2$. 
Build a tube (waveguide)  in~$\mathbb R^3$ by properly moving the region~$S$ along~$r(s)$; at each 
point~$r(s)$ the cross-section region $S$ may present a (continuously differentiable) rotation angle $\alpha(s)$.
Suppose that $\alpha(s)$ is  $L$-periodic.
For $\var > 0$ small enough, one can realize this same construction with the region $\var S$ and so obtaining a thin waveguide which is denoted by  $\Omega_\var$.

Let $-\Delta_{\Omega_\varepsilon}^D$ be the Dirichlet Laplacian on $\Omega_\varepsilon$.
Conventionally, $-\Delta_{\Omega_\varepsilon}^D$ is the Friedrichs extension of the Laplacian operator
$-\Delta$ in $L^2(\Omega_\varepsilon)$ with domain $C_0^\infty(\Omega_\varepsilon)$.
Denote by $\lambda_0 > 0$  the first eigenvalue of the Dirichlet Laplacian $-\Delta_S^D$ in $S$.
Due to the geometrical characteristics of $S$, $\lambda_0$ is simple.
One of the main results of this work is

\begin{Theorem}\label{maintheoremintrod}
For each $E > 0$, there exists  $\varepsilon_E > 0$ so that the spectrum of
$-\Delta_{\Omega_\varepsilon}^D$ is absolutely continuous in the interval 
$[0, \lambda_0/\varepsilon^2 + E]$, for all $\varepsilon \in (0, \varepsilon_E)$.
\end{Theorem}

In \cite{bento}, the authors proved  this result 
considering the particular case where the cross section of $\Omega_\varepsilon$
is a ball ${\cal B}_\varepsilon =\{y \in \mathbb R^2: |y| < \varepsilon\}$ 
(this fact eliminates the twist effect).
Covering the case where $\Omega_\varepsilon$ can be simultaneously curved and twisted is our main
contribution on the theme.

Ahead, we summarize the main steps to prove Theorem \ref{maintheoremintrod}. 
In particular, we call attention to Theorem
\ref{reductionofdimension}
and Corollary \ref{asymptoticbehaviour}, which are
our main tools to generalize the result of \cite{bento}.
Then, we present the results
related to the existence and location of gaps in $\sigma(-\Delta_{\Omega_\varepsilon}^D)$.
Many details are omitted in this introduction but will be presented
in the next sections.

Fix a number $c > \|k^2/4\|_\infty$. Denote by $\Id$ the identity operator.
For technical reasons, we  start to study the operator
$-\Delta_{\Omega_\varepsilon}^D + c \, \Id$; see Section \ref{reductiondimsec}. 

A change of coordinates shows that $-\Delta_{\Omega_\varepsilon}^D + c \, \Id$ is unitarily equivalent to 
the operator
\begin{equation}\label{introdfirstchange}
T_\varepsilon \psi := -\frac{1}{\beta_\varepsilon}
(\partial_{sy}^R \beta_\varepsilon^{-1} \partial_{sy}^R) \psi 
- \frac{1}{\varepsilon^2 \beta_\varepsilon} {\rm div}( \beta_\varepsilon \nabla_y \psi) + c \, \psi,
\end{equation}
where 
\begin{equation}\label{partialxynotation}
\partial_{sy}^R \psi := \psi' + \langle \nabla _y \psi, R \,y \rangle (\tau+\alpha')(s),
\end{equation}
${\rm div}$ denotes the divergent of a vetor field in $S$,
$\psi':= \partial \psi / \partial s$,
$\nabla_y \psi := (\partial \psi/\partial y_1, \partial \psi/\partial y_2)$ and
$R$ is the rotation matrix $\left(\begin{array}{cc}
0 & -1 \\
1  & 0
\end{array}\right)$.
The domain
$\dom T_\varepsilon$ is a subspace of the Hilbert space $L^2(\mathbb R \times S, \beta_\varepsilon \ds \dy)$
where the measure $\beta_\varepsilon \ds\dy$ comes from the Riemannian metric (\ref{metric});
see Section \ref{changecoordinates} for the exact definition of $\beta_\varepsilon$ and details of this transformation.

Since the coefficients of $T_\varepsilon$ are periodic with respect to $s$, we utilize the Floquet-Bloch reduction under the Brillouin zone ${\cal C}:=[-\pi/L, \pi/L]$. More precisely, we show that
$T_\varepsilon$ is unitarily equivalent to the operator
$\int_{{\cal C}}^\oplus T_\varepsilon^\theta \, d \theta,$
where
\begin{equation}\label{fiberoperatort}
T_\varepsilon^\theta \psi := \frac{1}{\beta_\varepsilon} 
(-i\partial_{sy}^R + \theta) \beta_\varepsilon^{-1} (-i\partial_{sy}^R + \theta) \psi
- \frac{1}{\varepsilon^2 \beta_\varepsilon} {\rm div}(\beta_\varepsilon \nabla_y \psi) + c \, \psi.
\end{equation}
Now, the domain of $T_\varepsilon^\theta$ is a subspace of $L^2((0,L)\times S, \beta_\varepsilon \ds \dy)$
and, in particular, the functions in $\dom T_\varepsilon^\theta$ satisfy the boundary conditions
$\psi(0,y)=\psi(L,y)$ and $\psi'(0,y)=\psi'(L,y)$ in $L^2(S)$.
Furthermore,
each  $T_\varepsilon^\theta$ is self-adjoint.
See Lemma \ref{floquetdecomposition} in Section \ref{floquetsection} for this decomposition.

Each $T_\varepsilon^\theta$ has compact resolvent and is bounded from below. Thus, 
$\sigma(T_\varepsilon^\theta)$ is discrete.
Denote by $\{E_n(\varepsilon, \theta)\}_{n \in \mathbb N}$ the family of all  eigenvalues of $T_\varepsilon^\theta$ 
and by
$\{\psi_n(\varepsilon, \theta)\}_{n \in \mathbb N}$ the family of the corresponding normalized  eigenfunctions, i.e.,
\[T_\varepsilon^\theta \psi_n(\varepsilon, \theta) = E_n(\varepsilon, \theta) \psi_n(\varepsilon, \theta),
\quad n =1, 2, 3, \cdots, \quad \theta \in {\cal C}.\]

We have
\begin{equation}\label{sigmaband}
\sigma(-\Delta_{\Omega_\varepsilon}^D) = \cup_{n=1}^{\infty} \left\{E_n(\varepsilon, {\cal C})\right\},
\quad \hbox{where} \quad E_n(\varepsilon, {\cal C}) :=  
\cup_{\theta \in {\cal C}} \left\{E_n(\varepsilon, \theta)\right\};
\end{equation}
each  $E_n(\varepsilon, {\cal C})$ is called  $n$th band of 
$\sigma(-\Delta_{\Omega_\varepsilon}^D)$.

We begin with the following result.

\begin{Lemma}\label{analyticfamily}
$\{T_\varepsilon^\theta:\theta \in {\cal C}\}$ is a type A analytic family.
\end{Lemma}

This lemma ensures that the functions $E_n(\varepsilon, \theta)$ 
are real analytic in ${\cal C}$ (its proof 
is presented in  Section \ref{floquetsection}).

Another important point to prove Theorem  \ref{maintheoremintrod} is to know an
asymptotic behavior of the eigenvalues $E_n(\varepsilon, \theta)$ as $\varepsilon$ tends to $0$.
For this characterization, 
for each $\theta \in {\cal C}$, consider the one dimensional self-adjoint operator
\[T^\theta w := (-i \partial_s + \theta)^2 w + \left[
C(S) (\tau+\alpha')^2(s) + c - \frac{k^2(s)}{4} \right] w,\]
acting in $L^2(0,L)$, where the functions in $\dom T^\theta$ satisfy the conditions
$w(0)=w(L)$ and $w'(0)=w'(L)$.
The constant $C(S)$ depends on the cross section $S$ and is defined by 
(\ref{defC(S)}) in Section \ref{reductiondimsec}.

For simplicity, write $Q:=(0,L) \times S$.
Recall  $\lambda_0 > 0$ denotes  the first eigenvalue of the Dirichlet Laplacian $-\Delta_S^D$ in $S$.
Denote by
$u_0$ the corresponding normalized eigenfunction.
Consider the closed subspace ${\cal L} :=\{w(s) u_0(y): w \in L^2(0,L) \} \subset L^2(Q)$
and the unitary operator ${\cal V}_\varepsilon$ defined by 
(\ref{unitaryappendixtheo}) in Section \ref{reductiondimsec}.
Our main tool to find an asymptotic behavior for $E_n(\varepsilon, \theta)$, and then to conclude
Theorem \ref{maintheoremintrod}, is given by
\begin{Theorem}\label{reductionofdimension}
There exists a number $K > 0$ so that, for all $\varepsilon > 0$ small enough,
\begin{equation}\label{compresulintro}
\sup_{\theta \in {\cal C}} \left\{\left\|{\cal V}_\varepsilon^{-1}\left(T_\varepsilon^\theta - \frac{\lambda_0}{\varepsilon^2} \Id\right)^{-1} {\cal V}_\varepsilon 
- ((T^\theta)^{-1} \oplus {\bf 0} ) \right\|\right\} \leq K \, \varepsilon,
\end{equation}
where ${\bf 0}$ is the null operator on the subspace ${\cal L}^\perp$.
\end{Theorem}

The spectrum of $T^\theta$ is purely discrete; 
denote by $\kappa_n(\theta)$ its $n$th eigenvalue counted with multiplicity.
Let ${\cal K}$ be
a compact subset of ${\cal C}$ which contains an open interval and does not contain the points
$\pm \pi/L$ and $0$.
Given $E>0$, without lost of generality, we can suppose that,
for all $\theta \in {\cal K}$,
the spectrum of $T_\var^\theta$ below $E+\lambda_0/\var^2$ consists of exactly $n_0$ eigenvalues
$\{E_n(\var, \theta)\}_{n=1}^{n_0}$.
As a consequence of Theorem \ref{reductionofdimension},

\begin{Corollary}\label{asymptoticbehaviournew}
There exists $\varepsilon_{n_0} > 0$ so that, 
for all 
$\varepsilon \in (0, \varepsilon_{n_0})$,
\begin{equation}\label{asymptoticintroduction}
E_n(\varepsilon, \theta) = \frac{\lambda_0}{\varepsilon^2} + \kappa_n(\theta) + O(\varepsilon),
\end{equation}
holds for each $n=1,2, \cdots, n_0$, uniformly in ${\cal K}$.
\end{Corollary}

In \cite{bento} the authors found a similar approximation as in Theorem \ref{reductionofdimension} that also holds uniformly 
for $\theta$ in ${\cal K}$. 
However their results were proved with the assumption that the cross section was a ball ${\cal B}_\varepsilon$. In their 
proofs, they have used results of \cite{duclos} which do not seem to generalize easily to other cross sections. 
On the other hand, similar estimates to (\ref{compresulintro}) and
(\ref{asymptoticintroduction})
were proved in \cite{bmt, oliveira04, davidsedi} for a larger class of cross sections
than only balls,
but the results hold only in the case $\theta=0$. 
We stressed that in \cite{davidsedi} the convergence is established without assuming
the existence of a Frenet frame in the reference curve $r$.

With all these tools in hands, we have 

\vspace{0.3cm}
\noindent
{\bf Proof of Theorem \ref{maintheoremintrod}:}
Let $E >0$,  without loss of generality, we suppose that, for all $\theta \in {\cal K}$, the spectrum
of $T_\varepsilon^\theta$ below $E+\lambda_0/\var^2$ consists of exactly $n_0$ eigenvalues 
$\{E_n(\varepsilon, \theta)\}_{n=1}^{n_0}$. 
Lemma \ref{analyticfamily} ensures that $E_n(\varepsilon, \theta)$ are real analytic functions.
To conclude the  theorem, it remains to show that each
$E_n(\varepsilon, \theta)$ is nonconstant.

Consider the functions $\kappa_n(\theta)$, $\theta \in {\cal K}$.
By Theorem XIII.89 in \cite{reed}, they are  nonconstant.
By Corollary \ref{asymptoticbehaviour}, there exists
$\varepsilon_E > 0$  so that
(\ref{asymptoticintroduction}) holds true for $n=1,2, \cdots, n_0$, 
uniformly in $\theta \in {\cal K}$,
for all $\varepsilon \in (0, \varepsilon_E)$.
Note that $\varepsilon_E > 0$ depends on $n_0$, i.e.,
the thickness of the tube depends on the length
of the energies to be covered.
By Section XIII.16 in \cite{reed}, the conclusion follows.

\vspace{0.3cm}

We know that the spectrum of $-\Delta_{\Omega_\varepsilon}^D$  
coincides with the union of bands; see (\ref{sigmaband}).
It is natural to question  the existence of gaps in its structure.
This subject was studied in \cite{Yoshi}. In that work, by
considering  a curved waveguide in $\mathbb R^2$,
the author ensured the existence of at least one gap in the spectrum
of the Dirichlet Laplacian
and found its location.
In this work, we prove similar results for the operator $-\Delta_{\Omega_\varepsilon}^D$.

At first,
it is possible to organize the eigenvalues $\{E_n(\var, \theta)\}_{n \in \mathbb N}$ of  $T_\varepsilon^\theta$ 
in order to obtain a non-decreasing sequence. We keep the same notation and write
\[E_1(\var, \theta) \leq E_2(\var, \theta) \leq \cdots \leq E_n(\var, \theta) \cdots, \quad
\theta \in {\cal C}.\]
In this step 
the functions $E_n(\var, \theta)$ are 
continuous and piece-wise analytic in ${\cal C}$ (see Chapter $7$ in \cite{kato});
each $E_n(\var, {\cal C})$ is either a closed interval or a one point set.
In this case, similar to Corollary \ref{asymptoticbehaviournew}, we have

\begin{Corollary}\label{asymptoticbehaviour}
For each  $n_0 \in \mathbb N$, there exists $\varepsilon_{n_0} > 0$ so that, 
for all 
$\varepsilon \in (0, \varepsilon_{n_0})$,
\begin{equation}\label{asymptoticintroduction}
E_n(\varepsilon, \theta) = \frac{\lambda_0}{\varepsilon^2} + \kappa_n(\theta) + O(\varepsilon),
\end{equation}
holds for each $n=1,2, \cdots, n_0$, uniformly in ${\cal C}$.
\end{Corollary}

For simplicity of notation, write
\[V(s):= C(S) (\tau+\alpha')^2(s) + c - \frac{k^2(s)}{4}.\]

\begin{Theorem}\label{corgap}
Suppose that $V(s)$ is  not constant. Then, there exist  $n_1 \in \mathbb N$, $\varepsilon_{n_1+1} > 0$ and
$C_{n_1} > 0$ so that, for all $\varepsilon \in (0, \varepsilon_{n_1+1})$,
\begin{equation}\label{locbanddiri}
\min_{\theta \in {\cal C}} E_{n_1+1} (\varepsilon, \theta) 
-
\max_{\theta \in {\cal C}} E_{n_1} (\varepsilon, \theta) =
C_{n_1} + O(\varepsilon).
\end{equation}
\end{Theorem}

Theorem \ref{corgap} ensures  that at least one gap appears in the spectrum 
$\sigma(-\Delta_{\Omega_\varepsilon}^D)$ for
$\varepsilon > 0$ small enough.
Its proof is based on arguments of
\cite{borg, Yoshi}
and will be presented in Section \ref{sectionbandgaps}.


With the next result, it will be possible to find a location where (\ref{locbanddiri}) holds true. However,
some adjustments will be necessary.

For $\gamma > 0$, we use the scales 
\begin{equation}\label{scale}
k(s) \mapsto \gamma \, k(s), \quad
(\tau+\alpha')(s) \mapsto \gamma \, (\tau+\alpha')(s) \quad \hbox{and}
\quad c \mapsto \gamma^2 \, c.
\end{equation}
Thus, we obtain a  new region $\Omega_{\gamma, \varepsilon}$ and we consider
$-\Delta_{\Omega_{\gamma, \varepsilon}}^D$ instead of $-\Delta_{\Omega_\varepsilon}^D$. 
Denote by $T_{\gamma, \varepsilon}$ and $T_{\gamma, \varepsilon}^\theta$
the operators obtained by replacing (\ref{scale}) in (\ref{introdfirstchange})
and (\ref{fiberoperatort}), respectively.
Denote by $E_n(\gamma, \varepsilon, \theta)$ the $n$th 
eigenvalue of $T_{\gamma, \varepsilon}^\theta$ counted with multiplicity.

Expand the function $V(s)$ as a Fourier series, i.e.,
\[V(s) = \sum_{n=-\infty}^{+\infty} \frac{1}{\sqrt{L}} \nu_n e^{2\pi n i s/L}
\quad \hbox{in} \,\, L^2(0,L),\]
where the sequence $\{\nu_n\}_{n=-\infty}^{+\infty}$ is called Fourier coefficients of $V(s)$.
Since $V(s)$ is a real function, $\nu_n=\overline{\nu}_{-n}$, for all $n \in \mathbb Z$.
We have the following result.
\begin{Theorem}\label{locationgapsgaps}
Suppose that $V(s)$ is not constant, and let 
$n_2 \in \mathbb N$ so that $\nu_{n_2} \neq 0$.
Then, there exist $\gamma > 0$ small enough, $\varepsilon_{n_2+1} > 0$ and $C_{\gamma, n_2} > 0$ so that,
for all $\varepsilon \in (0, \varepsilon_{n_2+1})$,
\[\min_{\theta \in {\cal C}} E_{n_2+1} (\gamma, \varepsilon, \theta) 
-
\max_{\theta \in {\cal C}} E_{n_2} (\gamma,\varepsilon, \theta) =
C_{\gamma,n_2} + O(\varepsilon).\]
\end{Theorem}

As Theorem \ref{corgap}, the proof of Theorem  \ref{locationgapsgaps} is based on \cite{Yoshi}
and will be presented in Section \ref{locgaps}.

This work is written as follows. In Section \ref{changecoordinates}
we construct with details the tube $\Omega_\varepsilon$
where the Dirichlet Laplacian operator is considered.
In the same section, we realize a change of coordinates that allows us ``straight''
$\Omega_\varepsilon$, i.e., to work in the Hilbert space $L^2(\mathbb R \times S, \beta_\varepsilon \ds \dy)$.
In Section \ref{floquetsection} we perform the Floquet-Bloch decomposition and
prove  Lemma \ref{analyticfamily}.
Section \ref{reductiondimsec} is intended at proofs of Theorem 
\ref{reductionofdimension} and Corollary \ref{asymptoticbehaviour} (Corollary \ref{asymptoticbehaviournew} can be proven in 
a similar way and we omit its proof in this text).
Sections \ref{sectionbandgaps} and \ref{locgaps} are dedicated to the proofs of Theorems
\ref{corgap} and \ref{locationgapsgaps}, respectively.

A long the text, the symbol $K$ is used to denote different constants and it never depends on 
$\theta$.

\section{Geometry of the domain and change of coordinates}\label{changecoordinates}

Let $r : \mathbb R \to \mathbb R^3$ be a simple $C^3$ curve in~$\mathbb R^3$ parametrized by
its arc-length parameter~$s$. We suppose that $r$ is periodic, i.e., there exists $L>0$ 
and a nonzero vector $u$
so that
\[r(s+L)=u+r(s), \qquad \forall s \in \mathbb R.\]
The curvature of~$r$ at the position~$s$ is $k(s) := \| r''(s)\|$. We assume $k(s)>0$, for all $s \in \mathbb R$. 
Then, $r$ is endowed with the Frenet frame $\{T(s), N(s), B(s)\}$ given by the
tangent, normal and binormal vectors, respectively, moving along the curve and defined by
\[T = r'; \quad N = k^{-1} T'; \quad B = T \times N.\]
The Frenet equations are satisfied, that is,
\begin{equation}\label{frame}
\left(\begin{array}{c}
T' \\
N' \\
B'
\end{array}\right) =
\left( 
\begin{array}{ccc}
0 & k & 0 \\
-k & 0 & \tau \\
0 & -\tau & 0 
\end{array}
\right)
\left(\begin{array}{c}
T \\
N \\
B
\end{array}\right),
\end{equation}
where~$\tau(s)$ is the torsion of $r(s)$, actually defined by~\eqref{frame}. 
More generally, we can consider the case where $r$ has pieces of straight lines, i.e., $k = 0$ 
identically in these pieces. In this situation, the construction of a $C^2$ Frenet frame is 
described in Section 2.1 of \cite{ekholmk}. As another alternative, one can assume the Assumption 1 from $\cite{duclosfk}$.
For simplicity, we also denote by $\{ T(s), N(s), B(s)\}$  the Frenet frame in those cases.

Let $\alpha: \mathbb R \to \mathbb R$ be an $L$-periodic  and 
$C^1$ function so that $\alpha(0) = 0$, and $S$ an open, bounded,  connected and smooth (nonempty) subset of~$\mathbb R^2$. For~$\var > 0$ small enough and $y = (y_1, y_2) \in S$, write
$$x(s, y) = r(s) + \var  y_1 N_\alpha (s) + \var y_2 B_\alpha(s)$$
and consider the domain
$$\Omega_\var = \{x(s,y) \in \mathbb R^3: s \in \mathbb R, y=(y_1, y_2) \in S\},$$
where
\begin{eqnarray*}
N_\alpha(s) & : = &  \cos \alpha(s) N(s) + \sin \alpha(s) B(s), \\
B_\alpha(s) & := &  - \sin \alpha(s) N(s) + \cos \alpha(s) B(s).
\end{eqnarray*}

Hence, this tube 
$\Omega_\var$ 
is obtained by putting the region $\var S$ along the curve
$r(s)$, which is simultaneously rotated by an angle $\alpha(s)$ with respect to the
cross section at the position $s= 0$.

As already mentioned in the Introduction,
let $-\Delta_{\Omega_\varepsilon}^D$ be the Friedrichs extension of the Laplacian operator
$-\Delta$ in $L^2(\Omega_\var)$ with domain $C_0^\infty(\Omega_\var)$.

The next step is to perform a change of variables 
so that  $\Omega_\varepsilon$ is homeomorphic to the straight cylinder $\mathbb R \times S$.
Consider the mapping
$$\begin{array}{llll}
F_\var: & \mathbb R \times S & \to & \Omega_\var \\
     & (s,y)   & \mapsto     & r(s) + \var y_1 N_\alpha(s) + \var y_2 B_\alpha(s).
\end{array}$$
In the new variables,  the Dirichlet Laplacian 
$-\Delta_{\Omega_\varepsilon}^D$ will be unitarily equivalent to one operator
acting in $L^2(\mathbb R \times S, \beta_\varepsilon \ds \dy)$; see definition of $\beta_\varepsilon$ below. 
The price to be paid is a nontrivial Riemannian metric $G = G_\var^\alpha$  which is induced by~$F_\var$, i.e.,
\begin{equation}\label{metric}
G=(G_{ij}), \qquad G_{ij} = \langle e_i, e_j \rangle = G_{ji}, \qquad 1 \leq i, j \leq 3,
\end{equation}
where
$$e_1 = \frac{\partial F_\var}{\partial s}, \quad
e_2 = \frac{\partial F_\var}{\partial y_1},
\quad
e_3 = \frac{\partial F_\var}{\partial y_2}.$$

Some calculations show that in the Frenet frame
$$J :=
\left(\begin{array}{c}
e_1 \\
e_2 \\
e_3
\end{array}\right) =
\left(
\begin{array}{ccc}
\beta_\var   & -\var (\tau+\alpha') \langle z_\alpha^\perp, y \rangle & \var (\tau+\alpha') \langle z_\alpha, y \rangle \\
0    & \var \cos \alpha  & \var \sin \alpha \\
0    & -\var \sin \alpha   & \var \cos \alpha
\end{array}
\right),$$
where 
\begin{equation}\label{defbetarestric}
\beta_\var(s,y) := 1 - \var k(s) \langle z_\alpha, y \rangle,
\quad z_\alpha := (\cos \alpha, -\sin \alpha), \quad {\hbox{ and }} \quad
z_\alpha^\perp := (\sin \alpha, \cos \alpha).
\end{equation}
The inverse matrix  of $J$ is given by
$$\displaystyle J^{-1} =
\left(
\begin{array}{ccc}
1/\beta_\var   &   (\tau+\alpha') y_2/\beta_\var  &  -(\tau+\alpha') y_1/\beta_\var \\
0 & (1/\var) \cos \alpha   & -(1/\var)\sin \alpha   \\
0   & (1/\var) \sin \alpha   & (1/\var) \cos \alpha
\end{array}
\right).$$

Note that $JJ^t = G$ and $\det J = | \det G|^{1/2} = \var^2 \beta_\var$. Since $k$ is a bounded function, for $\var$ small 
enough, $\beta_\var$ does not vanish in $\mathbb R \times S$. Thus,  $\beta_\var> 0$ and~$F_\var$ is a local diffeomorphism. 
By requiring that $F_\var$ is injective (i.e., the tube is not self-intersecting), a global diffeomorphism is obtained.

Finally, consider the unitary transformation
$$\begin{array}{cccc}
{\cal J}_\var:     &      \LL^2(\Omega_\var) & \rightarrow  & \LL^2(\mathbb R \times S, \beta_\var \ds \dy) \\
            &      u                & \mapsto      & \var \,  u \circ F_\var
\end{array},$$
and recall the operator $T_\varepsilon$ given by (\ref{introdfirstchange}) in the Introduction.
After some straightforward calculations, we can show that
${\cal J}_\varepsilon (-\Delta_{\Omega_\varepsilon}^D) {\cal J}_\varepsilon^{-1} \psi = T_\varepsilon \psi$, where
$\dom T_\varepsilon= {\cal J}_\varepsilon (\dom (-\Delta_{\Omega_\varepsilon}^D))$. 
From now on, we start to study  $T_\varepsilon$.

\section{Floquet-Bloch decomposition}\label{floquetsection}

Since the coefficients  of $T_\varepsilon$ are periodic with respect to~$s$, in this section we
perform the Floquet-Bloch reduction over the Brillouin zone ${\cal C} =[-\pi/L, \pi/L]$.
For simplicity of notation, we write
$\Omega:=\mathbb R \times S$, 
\[{\cal H}_\varepsilon := L^2(\Omega, \beta_\varepsilon \ds \dy), \quad
\tilde{{\cal H}}_\varepsilon := L^2(Q, \beta_\varepsilon \ds \dy).\]
Recall that $Q=(0,L) \times S$.

\begin{Lemma}\label{floquetdecomposition}
There exists a unitary operator 
${\cal U}_\varepsilon: {\cal H}_\varepsilon \to 
\int_{{\cal C}}^\oplus \tilde{{\cal H}}_\varepsilon \, d\theta$, so that,
\[{\cal U}_\varepsilon \, T_\varepsilon \, {\cal U}_\varepsilon^{-1} = \int_{{\cal C}}^\oplus 
T_\varepsilon^\theta \, d\theta,\]
where
\[T_\varepsilon^\theta \psi := \frac{1}{\beta_\varepsilon} 
(-i\partial_{sy}^R + \theta) \beta_\varepsilon^{-1} (-i\partial_{sy}^R + \theta) \psi
- \frac{1}{\varepsilon^2 \beta_\varepsilon} {\rm div}(\beta_\varepsilon \nabla_y \psi) + c \, \psi,\]
and, 
\[ \dom T_\varepsilon^\theta = \{\psi \in H^2(Q): \psi(s,y) = 0 \,\, \hbox{on} \,\, \partial Q \backslash  
\left(\{0, L\} \times S\right),\]
\hspace{5cm}$\psi(L,y)=\psi(0,y) \,\, \hbox{in} \,\, L^2(S), \, \psi'(L,y)=\psi'(0,y) \,\, \hbox{in} \,\, L^2(S)\}.$
\vspace{0.3cm}

\noindent
Furthermore,
for each $\theta \in {\cal C}$,  $T_\varepsilon^\theta$ is self-adjoint.
\end{Lemma}

\begin{proof}
As in \cite{bento},
for $(\theta, s,y) \in {\cal C} \times Q$ define
\[
({\cal U}_\varepsilon f)(\theta, s,y) := \sum_{n \in \mathbb Z} 
\sqrt{\frac{L}{2\pi}} e^{-inL\theta-i\theta s} f(s+Ln,y).\]
This transformation is a modification of Theorem XIII.88 in \cite{reed}.
As a consequence, the domain of the fibers operators $T_\varepsilon^\theta$ keep the same.

With respect to the proof of this lemma,
a detailed proof for periodic strips in the plane can be found in \cite{Yoshi}.
The argument for periodic waveguides in $\mathbb R^3$ is analogous and will be omitted in this text.
\end{proof}

\begin{Remark}
{\rm
Although $T_\varepsilon^\theta$ acts in the Hilbert space $\tilde{{\cal H}}_\varepsilon$, the operator
$\partial_{sy}^R \psi$ in its definition has action given by (\ref{partialxynotation}) (see Introduction) and
$\beta_\varepsilon$ is given by (\ref{defbetarestric}) (see Section \ref{changecoordinates}). For simplicity, 
we keep the same notation.}
\end{Remark}

Now, we present the proof of Lemma \ref{analyticfamily} stated in the Introduction.

\vspace{0.3cm}
\noindent
{\bf Proof of Lemma \ref{analyticfamily}:}
For each $\theta \in {\cal C}$, write $T_\varepsilon^\theta = T_\varepsilon^0 + V_\varepsilon^\theta$,
where, for $\psi \in \dom T_\varepsilon^0$,
\begin{eqnarray*}
V_\varepsilon^\theta \psi & := & (T_\varepsilon^\theta - T_\varepsilon^0)\psi \\
& = &
(-2 i \theta/\beta_\varepsilon^2) \partial_{sy}^R \psi + 
\left[-i\theta (\partial_{sy}^R \beta_\varepsilon^{-1})/\beta_\varepsilon + 
\theta^2/\beta_\varepsilon^2\right]
\psi.
\end{eqnarray*}

We affirm that $V_\varepsilon^\theta$ is $T_\varepsilon^0$-bounded with
zero relative bound. In fact, denote 
$R_z=R_z(T_\varepsilon^0) = (T_\varepsilon^0 - z \Id)^{-1}$.
Take $z \in \mathbb C$ with ${\rm img}\,z \neq 0$.
Since all coefficients of $V_\varepsilon^\theta$ are bounded, there exists $K>0$, so that,
\begin{eqnarray*}
\|V_\varepsilon^\theta \psi\|_{\tilde{{\cal H}}_\varepsilon}^2 & = & 
\int_{Q} |V_\varepsilon^\theta \psi|^2 \beta_\varepsilon \dx \dy\\
& \leq &
K \, \left(\langle  \psi, T_\varepsilon^0 \psi \rangle_{\tilde{{\cal H}}_\varepsilon}  + 
\|\psi\|^2_{\tilde{{\cal H}}_\varepsilon} \right) \\
& \leq &
K \, \left(\langle R_z(T_\varepsilon^0 - z \Id) \psi, T_\varepsilon^0 \psi \rangle_{\tilde{{\cal H}}_\varepsilon}  
+  \|\psi\|^2_{\tilde{{\cal H}}_\varepsilon} \right)\\
&\leq &
K \, \left(
\langle R_z T_\varepsilon^0 \psi, T_\varepsilon^0 \psi \rangle_{\tilde{{\cal H}}_\varepsilon}  
+  |z| \langle \psi, R_{\overline{z} } T_\varepsilon^0 \psi \rangle_{\tilde{{\cal H}}_\varepsilon} 
+ \|\psi\|^2_{\tilde{{\cal H}}_\varepsilon} \right) \\
& \leq &
K \, \left(\|R_z T_\varepsilon^0 \psi\|_{\tilde{{\cal H}}_\varepsilon} 
\|T_\varepsilon^0 \psi\|_{\tilde{{\cal H}}_\varepsilon} +
|z| \langle \psi, (\Id+\overline{z} R_{z}) \psi \rangle_{\tilde{{\cal H}}_\varepsilon} 
+  \|\psi\|^2_{\tilde{{\cal H}}_\varepsilon} \right) \\
& \leq &
K \, \left[ \|R_z\|_{\tilde{{\cal H}}_\varepsilon}  \|T_\varepsilon^0 \psi\|_{\tilde{{\cal H}}_\varepsilon}^2 + 
  \left( |z| +|z|^2 \|R_{z} \|_{\tilde{{\cal H}}_\varepsilon} + 1 \right) 
\|\psi\|_{\tilde{{\cal H}}_\varepsilon}^2 \right],
\end{eqnarray*}
for all  $\psi \in \dom T_\varepsilon^0$ and all $\theta \in {\cal C}$.
In the first inequality we use the Minkovski inequality and  the
property  
$ab\leq (a^2+b^2)/2$, for all $a,b \in \mathbb R$. 
In the third one, we used that $R_{\overline{z}} T_\varepsilon^0 = \Id + \overline{z} R_z$.

Since $\|R_z\|_{\tilde{{\cal H}}_\varepsilon} \to 0$, as ${\rm img}\,z \to \infty$, 
the affirmation is proven. So, the lemma follows.

\section{Proof of Theorem \ref{reductionofdimension} and 
Corollary \ref{asymptoticbehaviour}}\label{reductiondimsec}

This section is dedicated to prove Theorem \ref{reductionofdimension}.
Some steps are very similar to that in \cite{oliveira04} and require only an 
adaptation. Because this, most calculations will be omitted here.

Since $T_\varepsilon^\theta > 0$ is self-adjoint, there exists a 
closed sesquilinear form $t_\varepsilon^\theta > 0$, so that,
$\dom T_\varepsilon^\theta \subset \dom t_\varepsilon^\theta$
(actually, $\dom T_\varepsilon^\theta$ is a core of $\dom t_\varepsilon^\theta$) and
\[t_\varepsilon^\theta(\phi, \varphi)=\langle \phi, T_\varepsilon^\theta \varphi\rangle,
\quad \forall \phi \in \dom t_\varepsilon^\theta, \forall \varphi \in \dom T_\varepsilon^\theta;\]
see Theorem $4.3.1$ of \cite{livrocesar}.

For $\varphi \in \dom T_\varepsilon^\theta$, the quadratic form 
$t_\varepsilon^\theta(\varphi) :=t_\varepsilon^\theta(\varphi, \varphi)$ acts as
\[t_\varepsilon^\theta(\varphi)  =
\int_Q
\frac{1}{\beta_\varepsilon} \left|\left(-i \partial_{sy}^R + \theta \right)\varphi \right|^2  \ds \dy 
+
\int_Q \frac{\beta_\varepsilon}{\varepsilon^2}  |\nabla_y \varphi|^2 \ds \dy +
c \int_Q  \beta_\varepsilon |\varphi|^2 \ds \dy.\]

We are interested in studying $t_\varepsilon^\theta(\varphi)$ for $\varepsilon > 0$ small enough.
However,  it is necessary to control the term 
$(1/\varepsilon^2)\int_Q \beta_\varepsilon |\nabla_y \varphi|^2 \ds \dy$,
as $\varepsilon \to 0$.
Since it is related to the transverse oscillations in the waveguide,
we make this in the following way.
As already mentioned in the Introduction, let $u_0$ be the eigenfunction associated with the first eigenvalue $\lambda_0$
of the Dirichlet Laplacian $-\Delta_S^D$ in $S$, i.e.,
\[-\Delta_S^D u_0 = \lambda_0 u_0, \quad
u_0 \geq 0, \quad \int_S |u_0|^2 \dy=1, \quad \lambda_0 > 0.\]
Due to the geometrical characteristics of $S$, $\lambda_0$ is a simple eigenvalue.
We consider the quadratic form
\begin{eqnarray*}
t_\varepsilon^\theta(\varphi) - \frac{\lambda_0}{\varepsilon^2} \|\varphi\|^2_{\tilde{{\cal H}}_\varepsilon}
& = &
\int_Q
\frac{1}{\beta_\varepsilon} \left|\left(-i \partial_{sy}^R + \theta \right)\varphi \right|^2  \ds \dy \\
& + &
\int_Q \frac{\beta_\varepsilon}{\varepsilon^2} \left( |\nabla_y \varphi|^2 - \lambda_0 |\varphi|^2 \right) \ds \dy +
c \int_Q  \beta_\varepsilon |\varphi|^2 \ds \dy,
\end{eqnarray*}
$\varphi \in \dom T_\varepsilon^\theta$.
The subtraction of $(\lambda_0/\varepsilon^2) \int_Q \beta_\varepsilon |\varphi|^2 \ds \dy$ 
is intended to control the divergence
of the transverse oscillations, as $\varepsilon \to 0$
(see a detailed discussion in Section $1$ of \cite{crospec}).

An important point is that, for each $\varphi \in \dom T_\varepsilon^\theta$, 
\[\int_S \frac{\beta_\varepsilon}{\var^2}
\left( |\nabla_y \varphi|^2 - \lambda_0 |\varphi|^2 \right) \\dy \geq \gamma_\varepsilon(s)
\int_S |\varphi|^2 \dy,
\quad \hbox{a.e. s},\]
where $\gamma_\varepsilon(s) \to - k^2(s)/4$ uniformly, as $\varepsilon \to 0$.
The proof of this inequality can be found in \cite{bmt}.
As a consequence, since $\|k^2/4\|_\infty < c$, zero belongs to the resolvent set 
$\rho\left(T_\varepsilon^\theta - (\lambda_0/\varepsilon^2) \Id\right)$, for all $\varepsilon > 0$ small enough.

Now, define the unitary operator
\begin{equation}\label{unitaryappendixtheo}
\begin{array}{cccc}
{\cal V}_\varepsilon: & L^2(Q) & \to  & \tilde{{\cal H}}_\varepsilon \\
               &  \psi             & \to  &   \psi/ \beta_\varepsilon^{1/2}
\end{array}.
\end{equation}
With this transformation, we 
start 
to work in $L^2(Q)$ with the usual measure of $\mathbb R^3$.
Namely, consider the quadratic form
\[b_\varepsilon^\theta(\psi)  := 
t_\varepsilon^\theta({\cal V}_\varepsilon ^\theta\psi)
-\frac{\lambda_0}{\varepsilon^2} \| {\cal V}_\varepsilon ^\theta\psi \|^2_{\tilde{{\cal H}}_\varepsilon},\]
defined on the subspace
$\dom b_\varepsilon^\theta := {\cal V}_\varepsilon^{-1} (\dom T_\varepsilon^\theta) \subset L^2(Q)$.
One can show 
\begin{eqnarray*}
b_\varepsilon^\theta(\psi) &  =  & 
\int_Q \frac{1}{\beta^2_\varepsilon} 
\left|-i \left[\partial_{sy}^R \psi +  \beta_\varepsilon^{1/2} (\partial_{sy}^R \beta_\varepsilon^{-1/2})
\psi \right] + \theta \psi \right|^2 \ds \dy \\
& + &
\int_Q
\frac{1}{\varepsilon^2}\left(|\nabla_y \psi|^2
- \lambda_0 |\psi|^2 \right) \ds \dy -
\int_Q  \frac{k^2(s)}{4 \beta_\varepsilon^2} |\psi|^2 \ds \dy
+ c \int_Q  |\psi|^2 \ds \dy.
\end{eqnarray*}
The details of the calculations in this change of coordinates can be found 
in  Appendix A of \cite{oliveira04}.

Denote by $B_\varepsilon^\theta$ the self-adjoint operator
associated with the closure $\overline{b}_\varepsilon^\theta$ of the quadratic form $b_\varepsilon^\theta$.
Actually, $\dom B_\varepsilon^\theta \subset \dom \overline{b}_\varepsilon^\theta$ and
\[{\cal V}_\varepsilon^{-1} \left( T_\varepsilon^\theta - \frac{\lambda_0}{\varepsilon^2} \Id\right) 
{\cal V}_\varepsilon = B_\varepsilon^\theta.\]

By replacing the global multiplicative factor $\beta_\varepsilon$  by $1$ in the first and third
integral in the expression of $b_\varepsilon^\theta(\psi)$, 
we arrive now at the  quadratic form
\begin{eqnarray*}
d_\varepsilon^\theta (\psi) &  :=  &
\int_Q 
\left|-i \left[\partial_{sy}^R \psi +  \beta_\varepsilon^{1/2} (\partial_{sy}^R \beta_\varepsilon^{-1/2})
\psi \right]+ \theta \psi \right|^2 \ds \dy \\
& + &
\int_Q \frac{1}{\varepsilon^2}
\left(|\nabla_y \psi|^2 - \lambda_0 |\psi|^2 \right) \ds \dy -
\int_Q  \frac{k^2(s)}{4} |\psi|^2 \ds \dy
+ c \int_Q  |\psi|^2 \ds \dy ,
\end{eqnarray*}
$\dom d_\varepsilon^\theta = \dom b_\varepsilon^\theta$.
Again, denote by $D_\varepsilon^\theta$ the self-adjoint operator associated with
the closure $\overline{d}_\var^\theta$ of the quadratic form $d_\varepsilon^\theta$.
We have $\dom D_\varepsilon^\theta = \dom B_\varepsilon^\theta$ and
$0 \in \rho(B_\varepsilon^\theta) \cap \rho(D_\varepsilon^\theta)$, for all
$\varepsilon > 0$ small enough.

To simplify the calculations ahead, we have the following result.

\begin{Theorem}\label{technicaltheorem}
There exists a number $K > 0$, so that, for all $\varepsilon > 0$ small enough,
\[\sup_{\theta \in {\cal C}} \left\{\|(B_\varepsilon^\theta)^{-1} - (D_\varepsilon^\theta)^{-1}  \| \right\}
\leq K \, \varepsilon.\]
\end{Theorem}

The main point in this theorem is  that $\beta_\varepsilon \to 1$ uniformily as $\varepsilon \to 0$.
Its proof is quite similar to the proof of Theorem $3.1$ in \cite{verri01} and will not be
presented here.

Consider the closed subspace ${\cal L} :=\{w(s)u_0(y): w \in L^2(0,L)\}$
of the Hilbert space $L^2(Q)$.
Take the 
orthogonal decomposition
\begin{equation}\label{decomporthog}
L^2(Q) = {\cal L} \oplus {\cal L}^\perp.
\end{equation}

For $\psi \in \dom D_\varepsilon^\theta$, we can write $\psi(s,y) = w(s) u_0(y) + \eta(s,y)$,
with $w \in H^2(0,L)$ and 
$\eta \in D_\varepsilon^\theta \cap {\cal L}^\perp$.
Furthermore, $w(0)=w(L)$.

Define
\begin{equation}\label{defC(S)}
C(S):= \int_S |\langle \nabla_y u_0, Ry \rangle|^2 \dy \geq 0.
\end{equation}
Note that $C(S)=0$ if, and only if, $S$ is radial.

Recall $V(s) = C(S) (\tau+\alpha')^2(s) + c - k^2(s)/4$
and  the one dimensional operator 
\[T^\theta w = (-i\partial_s + \theta)^2 w + V(s) w,\]
mentioned in the Introduction. Take 
$\dom T^\theta = \{w \in L^2(0,L): wu_0 \in \dom D_\varepsilon^\theta\} =
\{w \in H^2(0,L): w(0)=w(L), w'(0)=w'(L)\}$.
In this domain, $T^\theta$ is self-adjoint and, since $\|k^2/4\|_\infty < c$, 
$0 \in \rho(T^\theta)$.

Denote by $t^\theta(w)$ the  quadratic form associated with $T^\theta$. 
For $w \in \dom T^\theta$, 
\[t^\theta(w) = 
\int_0^L \Big{[} |(-i \partial_s + \theta ) w|^2 
+ V(s) |w|^2 \Big{]} \ds.\]

\vspace{0.3cm}
\noindent
{\bf Proof of Theorem \ref{reductionofdimension}:}
The proof  is separated in two steps.

\vspace{0.3cm}
\noindent
{\bf Step I.}
Define the one dimensional quadratic form
\[s_\varepsilon^\theta(w) :=  d_\varepsilon^\theta(wu_0)
=
\int_0^L \Big{[} |(-i \partial_s + \theta ) w|^2 
+ \left( W(s) + c + g_\varepsilon(s) \right) |w|^2 \Big{]} \ds,\]
$\dom s_\varepsilon^\theta = \dom T^\theta$,
where
\[g_\varepsilon(s) = \int_S   \left\{
\beta_\varepsilon (\partial_{sy}^R \beta_\varepsilon^{-1/2})^2
- \left[ \beta_\varepsilon^{1/2} (\partial_{sy}^R \beta_\varepsilon^{-1/2}) \right]' \right\} |u_0|^2 \dy
\in L^\infty(0,L).
\]

Actually, $s_\varepsilon^\theta$ is the restriction of $d_\varepsilon^\theta$ on the
subspace $\dom T^\theta = \dom D_\varepsilon^\theta \cap {\cal L}$.

Denote by $S_\varepsilon^\theta$ the self-adjoint operator associated 
with the closure $\overline{s}_\var^\theta$ of the quadratic form $s_\var^\theta$.
We have $\dom S_\varepsilon^\theta = \dom T^\theta \subset \dom \overline{s}_\varepsilon^\theta$.

Recall the definition of $\beta_\varepsilon$ by (\ref{defbetarestric}) in Section
\ref{changecoordinates}.
Some calculations show that 
\begin{equation}\label{convunigrest}
|g_\varepsilon (s)| \leq K \, \varepsilon, \quad \forall s \in (0,L),
\end{equation}
for some $K > 0$.
This fact and  the condition $\|k^2/4\|_\infty < c$ imply
$0 \in \rho(S^\theta_\varepsilon)$, for all $\varepsilon > 0$ small enough.

Let ${\bf 0}$ be the null operator on the subspace ${\cal L}^\perp$.
In this step, we are going to show that there exists  $K > 0$, so that, for all $\varepsilon > 0$ small enough,
\begin{equation}\label{firststepproof}
\sup_{\theta \in {\cal C}} \left\{ \|(D_\varepsilon^\theta)^{-1} - ((S_\varepsilon^\theta)^{-1} \oplus {\bf 0}) \| \right\}
\leq K \, \varepsilon.
\end{equation}

Due to the decomposition  (\ref{decomporthog}), for $\psi \in \dom D_\varepsilon^\theta$,
\[
\psi(s,y) = w(s) \, u_0(y) + \eta(s,y), \quad w \in \dom T^\theta, \quad \eta \in 
\dom D_\varepsilon^\theta \cap {\mathcal L}^\perp.
\]
Thus, $d_\varepsilon^\theta(\psi)$ can be rewritten 
as
\[
d_\varepsilon^\theta(\psi) = s_\varepsilon^\theta(w) +  d_\varepsilon^\theta(wu_0, \eta) 
+ d_\varepsilon^\theta(\eta, wu_0) + d_\varepsilon^\theta(\eta).
\]

We need to check that there are $c_0> 0$ and functions $0 \leq q(\varepsilon), 0 \leq p(\varepsilon)$
and $c(\varepsilon)$ 
so that $s_\varepsilon^\theta (w)$, $d_\varepsilon^\theta (\eta)$ and 
$d_\varepsilon^\theta (w, \eta)$
satisfy the following
conditions:
\begin{equation}\label{firstcondsol}
s_\varepsilon^\theta ( w ) \geq c(\varepsilon) \|wu_0\|_{\LL^2(Q)}^2, \quad \forall w \in \dom T^\theta, \quad c(\varepsilon) \geq c_0 > 0;
\end{equation}
\begin{equation}\label{secondcondsol}
d_\varepsilon^\theta (\eta) \geq p(\varepsilon) \| \eta\|_{\LL^2(Q)}^2, \quad 
\forall \eta \in \dom D_\varepsilon^\theta \cap {\mathcal L}^{\perp};
\end{equation}
\begin{equation}\label{thirdcondsol}
|d_\varepsilon^\theta ( w, \eta)|^2 \leq 
q(\varepsilon)^2\, s_\varepsilon^\theta ( w )\, d_\varepsilon^\theta (\eta), \quad
\forall \psi \in \dom D_\varepsilon^\theta;
\end{equation}
and with
\begin{equation}\label{soly}
p(\varepsilon) \to \infty, \quad c(\varepsilon)=O(p(\varepsilon)), \quad q(\varepsilon) \to 0 \quad {\rm as} \quad \varepsilon \to 0.
\end{equation}
Thus, Proposition~3.1 in~\cite{solomyak} guarantees that, for~$\varepsilon > 0$ small enough,
\[\sup_{\theta \in {\cal C}} \left\{
\| (D_\varepsilon^\theta )^{-1} - ((S_\varepsilon^\theta)^{-1} \oplus {\bf 0} ) \| \right\}
\leq p(\varepsilon)^{-1} + K \, q(\varepsilon) \, c(\varepsilon)^{-1},
\]
for some $K > 0$.
We highlight that the main point in this proof
is to get functions $c(\varepsilon), p(\varepsilon)$ and $q(\varepsilon)$ that do not depend on $\theta$.

Since $\|k^2/4\|_\infty < c$ and $g_\varepsilon(s) \to 0$ uniformly, 
there exists $c_1 > 0$, so that,
\[s_\varepsilon^\theta(w) \geq c_1 \int_0^L |w|^2 \ds = c_1 \|wu_0\|_{L^2(Q)},
\quad \forall w \in \dom T^\theta,\]
for all $\varepsilon > 0$ small enough.
We pick up $c(\varepsilon):=c_1$.

Let $\lambda_1 > \lambda_0$ the second eigenvalue of the Dirichlet Laplacian operator in $S$. The Min-Max
Principle ensures that
\[\int_S \left(|\nabla_y \eta|^2 - \lambda_0 |\eta|^2   \right) \dy \geq
(\lambda_1-\lambda_0) \int_S |\eta|^2  \dy,
\quad \hbox{a.e. s}, \quad \forall \eta \in \dom D_\varepsilon^\theta \cap {\cal L}^\perp.\]
Thus,
\[d_\varepsilon^\theta(\eta) \geq \frac{(\lambda_1-\lambda_0)}{\varepsilon^2} \int_Q |\eta|^2 \ds \dy,
\quad \forall \eta \in \dom D_\varepsilon^\theta \cap {\cal L}^\perp.\]
Just to take $p(\varepsilon):= (\lambda_1-\lambda_0)/\varepsilon^2$.

The proof of inequality (\ref{thirdcondsol})  is very similar to that in Appendix B in \cite{oliveira04}.
Again, it will be omitted here. One can show 
\[|d_\varepsilon^\theta ( w, \eta)|^2 \leq 
K\, \varepsilon^2\, s_\varepsilon^\theta ( w )\, d_\varepsilon^\theta (\eta), \quad
\forall \psi \in \dom D_\varepsilon^\theta,\]
for some $K> 0$.
Take $q(\varepsilon) := \sqrt{K} \, \varepsilon$.
Since the conditions
(\ref{firstcondsol}), (\ref{secondcondsol}), (\ref{thirdcondsol})
and 
(\ref{soly}) are satisfied,   (\ref{firststepproof}) holds true.

\vspace{0.3cm}
\noindent
{\bf Step II.}
By (\ref{convunigrest}), for all $\varepsilon > 0$ small enough,
\[|s_\varepsilon^\theta(w) - t^\theta(w)| \leq \|g_\varepsilon\|_\infty \int_0^L |w|^2 \ds \leq
K \, \varepsilon \int_0^L |w|^2 \ds,\quad \forall  w \in \dom T^\theta, \forall \theta \in {\cal C}.\]

By Theorem 3 in \cite{oliveira01}, for all $\varepsilon > 0$ small enough,
\[\sup_{\theta \in {\cal C}} \left\{ \|(S_\varepsilon^\theta)^{-1} - (T^\theta)^{-1} \|\right\} \leq K \, \varepsilon.\]

Taking into account Theorem \ref{technicaltheorem} and the Steps I and II,
we conclude the proof of Theorem \ref{reductionofdimension}.

\begin{Remark}\label{endremark}
{\rm
Let $(h_\varepsilon)_\varepsilon$, $(m_\varepsilon)_\varepsilon$ be two sequences of  positive and closed 
sesquilinear forms in the Hilbert space ${\cal H}$ with
$\dom h_\varepsilon = \dom m_\varepsilon = {\cal D}$, for all $\varepsilon > 0$.
Denote by $H_\varepsilon$ and $M_\varepsilon$ the self-adjoint operators associated with 
$h_\varepsilon$ and $m_\varepsilon$, respectively. Suppose that there exists $\zeta > 0$, so that,
$h_\varepsilon, m_\varepsilon > \zeta$, for all $\varepsilon > 0$, and
\begin{equation}\label{remarkreduction}
|h_\varepsilon(\varphi)-m_\varepsilon(\varphi)| \leq j(\varepsilon) \, m_\varepsilon(\varphi),
\quad \forall \varphi \in {\cal D},
\end{equation}
with $j(\varepsilon) \to 0$, as $\varepsilon \to 0$.
Theorem 3 in \cite{oliveira01} implies that there exists a number $K>0$, so that,
for all $\varepsilon > 0$ small enough,
\begin{equation}\label{concthedomoper}
\|H_\varepsilon^{-1} - M_\varepsilon^{-1} \| \leq K \, j(\varepsilon).
\end{equation}

Suppose that $\dom H_\varepsilon = \dom M_\varepsilon =: \tilde{\cal D}$ and that the condition 
(\ref{remarkreduction}) is satisfied for all $\varphi \in \tilde{{\cal D}}$.
By applying the 
same proof of \cite{oliveira01}, the inequality (\ref{concthedomoper}) holds true.

The same idea can be applied in Proposition~3.1 in~\cite{solomyak}. 
Because of this, in this section, when working with quadratic forms we have  
restricted the study to their actions in the domains of their respective
associated self-adjoint  operators.

}
\end{Remark}

\noindent
{\bf Proof of Corollary \ref{asymptoticbehaviour}:}
Denote by $\lambda_n(\varepsilon, \theta) := E_n(\varepsilon, \theta) - (\lambda_0/\varepsilon^2)$.
Theorem \ref{reductionofdimension} 
in the Introduction and Corollary 2.3 of \cite{gohberg}
imply
\begin{equation}\label{aproxinverse}
\left| \frac{1}{\lambda_n(\varepsilon, \theta)} - \frac{1}{\kappa_n(\theta)} \right| \leq K \, \varepsilon,
\quad \forall n \in \mathbb N, \, \forall \theta \in {\cal C},
\end{equation}
for all $\varepsilon > 0$ small enough.
Then,
\[\left|\lambda_n(\varepsilon, \theta) - k_n(\theta)\right| \leq 
K \, \varepsilon \, |\lambda_n(\varepsilon, \theta)| \, |k_n(\theta)|, \quad \forall n \in \mathbb N, \, \forall \theta \in {\cal C},\]
for all $\varepsilon > 0$ small enough.

A proof similar to that of Lemma \ref{analyticfamily} shows that
$\{T^\theta: \theta \in {\cal C}\}$ is a type $A$ analytic family. 
Thus,  the functions $k_n(\theta)$ are continuous in ${\cal C}$ and consequently bounded. 
This fact and the inequality (\ref{aproxinverse}) ensure that, for each $\tilde{n}_0 \in \mathbb N$, there exists
$K_{\tilde{n}_0}> 0$, so that,
\[|\lambda_{\tilde{n}_0}(\varepsilon, \theta) | \leq K_{\tilde{n}_0}, \quad \forall \theta \in {\cal C},\]
for all $\varepsilon > 0$ small enough.

Finally, for each $n_0 \in \mathbb N$, there exists $K_{n_0} > 0$ so that
\[\left|\lambda_n(\varepsilon, \theta) - k_n(\theta)\right| \leq 
K_{n_0} \, \varepsilon, \quad n=1, 2 \cdots, n_0, \forall \theta \in {\cal C},\]
for all $\varepsilon > 0$ small enough.

\section{Existence of band gaps; proof of Theorem \ref{corgap}}\label{sectionbandgaps}

Again, recall $V(s) = C(S) (\tau+\alpha')^2(s) + c - k^2(s)/4$ and consider the one dimensional operator 
\[Tw = - w'' + V(s) w , \quad \dom T = H^2(\mathbb R).\]

We have denoted by $\kappa_n(\theta)$ the $n$th eigenvalue (counted with multiplicity) of the operator $T^\theta$.
Each $\kappa_n(\theta)$ is a continuous function in ${\cal C}$.
By Chapter XIII.16 in \cite{reed}, we have the following properties:

\vspace{0.3cm}
\noindent
(a) $\kappa_n(\theta) = \kappa_n(-\theta)$, for all $\theta \in {\cal C}$, $n=1,2,3, \cdots$.

\vspace{0.3cm}
\noindent
(b) For $n$ odd (resp. even), $\kappa_n(\theta)$ is strictly  monotone increasing (resp. decreasing) as $\theta$
increases from $0$ to $\pi/L$. In particular,
\[\kappa_1(0) < \kappa_1(\pi/L) \leq \kappa_2(\pi/L) < \kappa_2(0) \leq \cdots
\leq \kappa_{2n-1}(0) < \kappa_{2n-1}(\pi/L)  \]
\[ \leq \kappa_{2n}(\pi/L) < \kappa_{2n}(0) \leq \cdots.\]

For each $n=1,2,3,\cdots$, define
\[ B_n := \left\{
\begin{array}{cc}
\left[ \kappa_n(0), \kappa_n(\pi/L) \right],  & \hbox{for} \, \, \,  n \,\,\, \hbox{odd}, \\
\left[ \kappa_n(\pi/L), \kappa_n(0) \right],  & \hbox{for} \, \, \, n \,\, \, \hbox{even},
\end{array}\right.
\]
and
\[ G_n := \left\{
\begin{array}{l}
\left(\kappa_n(\pi/L), \kappa_{n+1}(\pi/L) \right),  \,\,\,
\hbox{for} \, \, \,  n \,\,\, \hbox{odd so that} \,\,\, \kappa_n(\pi/L) \neq \kappa_{n+1}(\pi/L), \\
\left( \kappa_n(0), \kappa_{n+1}(0) \right),  \,\,\,
\hbox{for} \, \, \, n \,\, \, \hbox{even so that} \,\,\, \kappa_n(0) \neq \kappa_{n+1}(0), \\
\emptyset,  \,\,\,  \hbox{otherwise}. 
\end{array}
\right.
\]

By Theorem XIII.90 in \cite{reed}, one has $\sigma(T) = \cup_{n=1}^{\infty} B_n$ where
$B_n$ is called the $j$th band of $\sigma(T)$, and
$G_n$ the gap of $\sigma(T)$ if $B_n\neq \emptyset$.

Corollary \ref{asymptoticbehaviour} implies that for each $n_0 \in \mathbb N$, there exists $\varepsilon_{n_0} > 0$ so that,
for all $\varepsilon \in (0, \varepsilon_{n_0})$,
\[\max_{\theta \in {\cal C}} E_n (\varepsilon, \theta) =
\left\{
\begin{array}{l}
\lambda_0/\varepsilon^2 + \kappa_n(\pi/L) + O(\varepsilon), \,\,\, \hbox{for} \,\,\, n \,\,\, \hbox{odd}, \\
\lambda_0/\varepsilon^2 + \kappa_n(0) + O(\varepsilon), \,\,\, \hbox{for} \,\,\, n \,\,\, \hbox{even},
\end{array}
\right.
\]
and
\[\min_{\theta \in {\cal C}} E_n (\varepsilon, \theta) =
\left\{
\begin{array}{l}
\lambda_0/\varepsilon^2 + \kappa_n(0) + O(\varepsilon), \,\,\, \hbox{for} \,\,\, n \,\,\, \hbox{odd}, \\
\lambda_0/\varepsilon^2 +  \kappa_n(\pi/L) + O(\varepsilon),\,\,\, \hbox{for} \,\,\, n \,\,\, \hbox{even},
\end{array}
\right.
\]
hold for each $n=1,2,\cdots,n_0$.
Thus, we have

\begin{Corollary}\label{corcorrected}
For each $n_0 \in \mathbb N$, there exists $\varepsilon_{n_0+1} > 0$ so that,
for all $\varepsilon \in (0, \varepsilon_{n_0+1})$, 
\[\min_{\theta \in {\cal C}} E_{n+1} (\varepsilon, \theta) 
-
\max_{\theta \in {\cal C}} E_n (\varepsilon, \theta) =
|G_n| + O(\varepsilon),\]
holds for each $n=1,2, \cdots, n_0$,
where $| \cdot |$ is the Lebesgue measure.
\end{Corollary}

Another important tool to prove Theorem \ref{corgap} is 
the following result due to Borg \cite{borg}.

\begin{Theorem} (Borg)
Suppose that $W$ is a real-valued, piecewise continuous
function on $[0,L]$. Let $\lambda_n^{\pm}$ be the $n$th eigenvalue of the following operator
counted with multiplicity respectively
\[-\frac{d^2}{ds^2} + W(s), \quad \hbox{in} \quad 
L^2(0,L),\]
with domain
\begin{equation}\label{domainborg}
\{w \in H^2(0,L); w(0)=\pm w(L),
w'(0)=\pm w'(L)\}.
\end{equation}

We suppose that
\[\lambda_n^+ = \lambda_{n+1}^+, \quad
\hbox{for all even} \,\, n,\]
and
\[\lambda_n^- = \lambda_{n+1}^-, \quad
\hbox{for all odd} \,\, n.\]
Then,
$W$ is constant on $[0, L]$.
\end{Theorem}

\noindent
{\bf Proof of Theorem \ref{corgap}:}
For each $\theta \in {\cal C}$, we define the unitary transformation 
$(u_\theta w)(s) = e^{-i\theta s} w(s)$. In particular, consider the
operators  $\tilde{T}^{0}:= u_{0} T^{0} u_{0}^{-1}$
and $\tilde{T}^{\pi/L} := u_{\pi/L} T^{\pi/L} u_{\pi/L}^{-1}$
whose eigenvalues are given by 
$\{\nu_n(0)\}_{n\in \mathbb N}$ and $\{\nu_n(\pi/L)\}_{n\in \mathbb N}$, respectively. Furthermore, the domains of these operators
are given by (\ref{domainborg});
$\tilde{T}^0$ (resp. $\tilde{T}^{\pi/L}$) is called operator with periodic (resp. antiperiodic) boundary conditions.

Since $V(s)$ is  not constant in $[0,L]$, by Borg's Theorem, without loss of generality,
we can affirm that there exists $n_1 \in \mathbb N$ so that
$\nu_{n_1}(0) \neq \nu_{n_1+1}(0)$.
Now, the result follows by  Corollary  \ref{corcorrected}.

\section{Location of band gaps; proof of Theorem \ref{locationgapsgaps}}\label{locgaps}

The proof of Theorem \ref{locationgapsgaps} is very similar 
to the proof of Theorem 1.3 in \cite{Yoshi}. Due to this reason, we present only some steps.
A more complete proof can be found in that work.

We begin with some technical details.
Let $W \in L^2(0,L)$ be a real function.
For $\mu \in \mathbb C$, consider the operators
\[T^+ w = - w''+ \mu \, W(s) w \quad \hbox{and} \quad 
T^- w = - w'' + \mu \, W(s) w,\]
with domains given by
\begin{eqnarray*}
\dom T^+ & =  & \{ w \in H^2(0,L): w(0)=w(L), w'(0)=w'(L)\},\\
\dom T^- & =  & \{ w \in H^2(0,L): w(0)=-w(L), w'(0)=-w'(L)\},
\end{eqnarray*}
respectively.

Denote by $\{l_n^+(\mu)\}_{n \in \mathbb N}$ and
$\{l_n^-(\mu)\}_{n \in \mathbb N}$ the eigenvalues of $T^+$ and $T^-$, respectively.
For $\mu \in \mathbb R$ and $n \in \mathbb N$, define
\[\delta_n^+(\mu) := l_{2n+1}^+(\mu)-l_{2n}^+(\mu) \quad
\hbox{and}  \quad
\delta_n^-(\mu) := l_{2n}^-(\mu)-l_{2n-1}^-(\mu).\]
Now, 
\[\delta_{2n-1}(\mu) := \delta_n^-(\mu) \quad
\hbox{and} \quad
\delta_{2n}(\mu) := \delta_n^+(\mu).\]

Let $\{\omega_m\}_{n=-\infty}^{n=+\infty}$ be the Fourier coefficients of $W(s)$.
More precisely, one can write
\[W(s) = \sum_{n=-\infty}^{+\infty} \frac{1}{\sqrt{L}}   \omega_n e^{2 n \pi i s /L} \quad
\hbox{in }\, L^2(0,L).\]
Since $W(s)$ is a real function, we have $\omega_n = \overline{\omega_{-n}}$, for all
$n \in \mathbb Z$.

The goal is to find an asymptotic behavior for $\delta_n(\mu)$, as
$\mu \to 0$, in terms of the Fourier coefficients of
$W(s)$. 

\begin{Theorem}\label{theoapp}
For each $n \in \mathbb N$,
\[\delta_n(\mu) = \frac{2}{\sqrt{L}} |\omega_n| |\mu| + O(|\mu|^2),
\quad \mu \to 0, \, \mu \in \mathbb R.\]
\end{Theorem}

A detailed proof of Theorem \ref{theoapp} can be find in \cite{Yoshi};
the main tool used by the author in the proof
is the analytic perturbation theorem due to Kato and Rellich (see \cite{kato}; Chapter VII and Theorem 2.6
in Chapter VIII).

Recall the definition of $T^\theta$ and  $E_n(\gamma,\varepsilon, \theta)$ in the Introduction.
For each $\theta \in {\cal C}$, define
\[T_\gamma^\theta w := -w'' + \gamma^2 \, V(s) w, \quad \dom T_\gamma^\theta=\dom T^\theta.\]
Denote by $\kappa_n(\gamma, \theta)$ the $n$th eigenvalue of $T_\gamma^\theta$ counted with
multiplicity. 
As in Section \ref{sectionbandgaps},
consider the bands
\[ G_n(\gamma) := \left\{
\begin{array}{l}
\left(\kappa_n(\gamma, \pi/L), \kappa_{n+1}(\gamma, \pi/L) \right),  \,\,\,
\hbox{for} \, \, \,  n \,\,\, \hbox{odd so that} \,\,\, \kappa_n(\gamma, \pi/L) \neq \kappa_{n+1}(\gamma,\pi/L), \\
\left( \kappa_n(\gamma, 0), \kappa_{n+1}(\gamma,0) \right),  \,\,\,
\hbox{for} \, \, \, n \,\, \, \hbox{even so that} \,\,\, \kappa_n(\gamma,0) \neq \kappa_{n+1}(\gamma,0), \\
\emptyset,  \,\,\,  \hbox{otherwise}. 
\end{array}
\right.
\]
and 
note that
$|G_n(\gamma)| = \delta_n(\gamma), \forall n \in \mathbb N$, if we consider $\mu = \gamma^2$ and $W(s)=V(s)$.

We have
\begin{Corollary}\label{cornamtw}
For each $n_3 \in \mathbb N$, there exist $\gamma> 0$ small enough and $\varepsilon_{n_3+1}> 0$ so that,
for all $\varepsilon \in (0, \varepsilon_{n_3+1})$,
\begin{equation}\label{cornameqai}
\min_{\theta \in {\cal C}} E_{n_3+1} (\gamma,\varepsilon, \theta) 
-
\max_{\theta \in {\cal C}} E_{n_3} (\gamma,\varepsilon, \theta) =
|G_{n_3}(\gamma)| + O(\varepsilon),
\end{equation}
holds for each $n=1,2,\cdots, n_3$, where $| \cdot |$ is the Lebesgue measure.
\end{Corollary}

\noindent
{\bf Proof of Theorem \ref{locationgapsgaps}:}
Recall that we  have denoted by
$\{\nu_n\}_{n=-\infty}^{n=+\infty}$ the Fourier coefficients of
$V(s)$. 
Since $V(s)$ is not constant,
there exists $n_2 \in \mathbb N$ so that $\nu_{n_2} \neq 0$.

By Theorem \ref{theoapp},
\[|G_{n_2}(\gamma)| = \frac{2}{\sqrt{L}} \gamma^2 |\nu_{n_2}| + O(\gamma^4),
\quad \gamma \to 0.\]

On the other hand, by Corollary \ref{cornamtw}, there exists $\varepsilon_{n_2+1} > 0$ so that,
for all $\varepsilon \in (0, \varepsilon_{n_2+1})$, (\ref{cornameqai}) holds true.
Then, by taking $C_{\gamma, n_2}:=|G_{n_2}(\gamma)| > 0$, theorem is proven.

\vspace{0.4cm}
\noindent
{\bf Acknowledgments}
\vspace{0.4cm}

\noindent
The authors would like to thank Dr. C\'esar R. de Oliveira and 
Dr. David Krej${\rm\check{c}}$\'i${\rm\check{r}}$ik 
for useful discussions.




\end{document}